\documentclass[pdflatex,sn-chicago]{sn-jnl}
\usepackage{graphicx}%
\usepackage{multirow}%
\usepackage{amsmath,amssymb,amsfonts}%
\usepackage{amsthm}%
\usepackage{mathrsfs}%
\usepackage[title]{appendix}%
\usepackage{xcolor}%
\usepackage{textcomp}%
\usepackage{manyfoot}%
\usepackage{booktabs}%
\usepackage{algorithm}%
\usepackage{algorithmicx}%
\usepackage{algpseudocode}%
\usepackage{listings}%

\theoremstyle{thmstyleone}%
\newtheorem{theorem}{Theorem}
\newtheorem{proposition}[theorem]{Proposition}%

\theoremstyle{thmstyletwo}%

\theoremstyle{thmstylethree}%

\begin{document}

\title[Existence of Physical Functions]{Experiments, Computability, and the Existence of Physical Functions}

\author*[1]{\fnm{Isaac} \sur{P\'erez Castillo}}\email{iperez@izt.uam.mx}

\affil*[1]{\orgdiv{Departamento de F\'isica}, \orgname{Universidad Autónoma Metropolitana-Iztapalapa}, \orgaddress{\street{San Rafael Atlixco 186}, \city{Ciudad de México}, \postcode{09340}, \state{Ciudad de México}, \country{Mexico}}}

\abstract{Experimental science usually relies on laboratory procedures that, after finitely many steps, terminate with numerical reports on physical quantities. This paper argues that such procedures can be understood as algorithmic once the protocol, background conditions, and reporting rules are fixed. Assuming an explicit physical Church--Turing bridge principle, a reproducible experiment therefore computes a map from admissible inputs to outputs, and the corresponding function exists in the sense appropriate to those outputs. Furthermore, computable analysis allows us to explain why this conclusion is compatible with finite-precision measurement since in this case what matters is a systematic approximation to a requested accuracy, not the production of exact real numbers in a single step. Neither  protocol dependence nor stochasticity undermines the existence claim. Rather, they specify which map is realized by a given protocol and what additional assumptions are required for stronger claims about a single protocol-independent quantity. The paper therefore separates three questions that are often conflated: whether the function exists, whether it is computable, and when results obtained under different protocols may be treated as measurements of the same quantity.}


\maketitle

\section{Introduction}
\label{sec:intro}
A colleague in chemistry once asked me, during a very interesting lunch discussion about his work, whether the solubility function exists. He was of course not questioning whether chemists can measure solubility, since laboratories do this routinely. His worry was conceptual. The reported value depends on solvent, temperature, pressure, equilibration criteria, sample preparation, and readout conventions, so it was not obvious to him what the supposed function could be. He was implicitly imagining functions in a much stronger mathematical sense--as continuous maps between already well-defined spaces-- and under that picture the lack of a clean domain and codomain made the very existence of a solubility function look doubtful. This work argues that the answer is nevertheless straightforward: once the experimental conditions and the reporting rule are fixed, the experiment itself can be understood as an algorithm that naturally shows that the function exists as it is computed.

In this scenario, the reason for the existence question does not require an inflated conception of what a function is. For the purposes of this paper, a function is simply a mapping that assigns one output to each admissible input. One does not need, at the outset, a protocol-independent law already presented as a continuous map between well-behaved spaces. One only needs enough structure to say what counts as an admissible input, what counts as an output and what experimental protocol connects them. Under fixed experimental conditions, that structure is already supplied by the actual practice of carrying out the experiment. Thus, the problem is not whether some abstract mapping can be written down strictly in set-theoretic terms. The problem is whether a concrete experimental procedure fixes a target sharply enough for measurement, explanation, and prediction.

The key to the argument is to look at what an experiment actually is. An experiment is a finite and reproducible sequence of steps carried out to obtain a numerical report. One prepares a sample, fixes the relevant conditions, executes a prescribed procedure and, after finitely many steps, records a result. In the context of computability theory, this is an effective procedure and hence an algorithmic one. Given an admissible input under fixed background conditions, it terminates -- it halts -- with an output. This is the structure needed for the existence claim. Once the bridge principle is added, the experiment is an algorithm that computes a map from inputs to reported outputs.

Of course the previous claim must rely on a thesis to pass from an effective laboratory procedure to the existence of a function. This paper adopts an explicit physical Church--Turing bridge principle connecting finitely specified, reproducible laboratory procedures with Turing-computable maps \citep{Turing1937ComputableNumbers,Copeland2020ChurchTuringSEP,Piccinini2011PhysicalCTT}. Under this thesis, the map realized by the fixed experiment is computable in the sense appropriate to its outputs.

Note that neither arguments about finite precision nor about the stochastic nature of experiments undermine this conclusion. Experimental practice returns finite digits, intervals, and uncertainty statements rather than exact real numbers, but computable analysis is built precisely for that setting. A real-valued quantity is computable when there is a procedure that produces approximations to any requested accuracy \citep{Weihrauch2000ComputableAnalysis,BrattkaHertlingWeihrauch2008Tutorial}. The relevant question is therefore not whether a laboratory protocol outputs an exact real number in a single run, but whether it supports systematic finite approximations under controlled refinement. The stochastic case requires a corresponding distributional formulation, introduced below.

This way of framing the issue is consistent with familiar lessons from measurement theory and metrology. What is measured depends on how the task is specified, modeled, and reported \citep{Bridgman1927Logic,KrantzLuceSuppesTversky1971FOM1,JCGM2002012,JCGM1002008}, and recent discussion has emphasized that the legitimacy of a measurement target cannot always be taken for granted independently of practice \citep{Zhao2023MeasuringNonexistent}. But those issues should enter as qualifications to the procedure, not as the burden of the opening argument. To state the point simply, the main claim is that a fixed experiment secures the existence of the relevant function once it is treated, under the bridge principle, as computing a map.

The paper does not claim that every idealized physical magnitude is computable, that every experimentally accessible quantity is protocol independent, or that a general physical Church--Turing thesis has thereby been established. Its claim is narrower and sharper. A reproducible experiment is a finite procedure for obtaining a number. For that reason, it already defines a function at the level of its reports, and under the stated bridge principle that function is computable. The next section, Section~\ref{sec:experiments-effective}, clarifies the concepts of functions and experiments. The following sections state the bridge principle and the existence of the computed map, explain the finite-precision setting, introduce protocol dependence and the stochastic case only where they become necessary, and finally delimit the scope and consequences of the result.

\section{Functions and Experiments as Effective Procedures}
\label{sec:functions-experiments}
\label{sec:experiments-effective}
Let us start with a minimal notion of function. A function is a rule that assigns outputs to admissible inputs. Nothing in this requires continuity, differentiability, topology or some privileged real-valued law already detached from the experimental procedure by which values are obtained. Of course, those additional structures may become important later, but they are not required to ask whether a laboratory practice gives rise to a map. What is required is a set of admissible inputs, a set of possible outputs, and a procedure that connects the former to the latter.

In an experimental setting, these elements are fixed by the laboratorist's practice. Indeed, the admissible input is not an arbitrary object in the world, but an object that has been prepared, described, and controlled under a protocol. The measurement output is not an unconstrained mathematical value, but a report produced according to a stated rule. Therefore, for a fixed protocol $P$, background conditions $C$ and reporting rule $R$, an admissible input encoded with a finite description $x$ will result in a record which will be converted into a finite report $y$ via a finite sequence of prescribed operations carried out under $C$. In this sense a laboratory procedure has the form of an input--output procedure. It may be useful to write this structure explicitly as:
\begin{equation}
  x \overset{E}{\longmapsto} r \overset{R}{\longmapsto} y \,. 
\end{equation}
Here $E$ denotes the execution of the experimental protocol $P$ under some background conditions $C$, while $R$ corresponds to the reporting rule. The point of writing down explicitly this map is neither to trivialize nor to idealize experiments, as they are very complex and elaborate endeavors, but rather to emphasize once again that a fixed experiment receives admissible inputs, runs through finitely many prescribed steps, terminates according to its stopping rule, and produces a finite report.

Note that termination is essential. A laboratory protocol is not simply a physical process that evolves in time. It also includes a rule for when the run is complete. The stopping rule may be an equilibration criterion, a fixed integration time, a specified number of replicates, a calibration threshold, or a quality-control condition. If the stopping rule is not met, the run may fail or the particular input may fall outside the set of admissible inputs of the protocol. For this reason the map to be analyzed is not a map on every mathematically imaginable input, but a map realized on the set of admissible inputs for which the protocol is specified and successfully terminates.

This way of formulating the experiment is consistent with standard metrological practice. The International Vocabulary of Metrology (VIM) treats the specification of the measurand as part of the measurement task \citep{JCGM2002012}, and the Guide to the Expression of Uncertainty in Measurement (GUM) treats the measurement result as inseparable from the model, input quantities, and reporting conventions used to obtain it \citep{JCGM1002008,JCGMGUM62020}. This is also consistent with structural accounts of measurement, according to which measurement is not a mere readout but a partly empirical and partly theoretical process organized by background knowledge and procedural structure \citep{MariCarboneGiordaniPetri2017}. It is understood and agreed upon that the protocol, background conditions, and reporting rule are not external decorations added after the fact. They are an integral part of what fixes the input--output procedure whose computability will be the central point to be considered later.

Going back to the anecdote with my chemist colleague, the solubility example illustrates the point. Once the solvent, temperature, pressure, equilibrium criterion, sample preparation, and reporting convention are fixed, the input is no longer the vague object "a substance" considered in abstraction. It is a finitely described admissible sample under specified conditions. The protocol determines what is done to that sample. The apparatus and readout produce a finite record. The reporting rules convert that record into a numerical report. At this stage, the important point is not yet that the resulting map is computable. The important point is that the experiment has a clear structure of a finite effective procedure from admissible inputs to reported outputs.

This formulation also avoids a common overstatement. The claim is not that every physical process computes simply because one can describe it mathematically via a function. A laboratory experiment has a much more disciplined structure. Its inputs are prepared according to rules, its operations are publicly specifiable, its stopping conditions are part of the protocol and its outputs are read through established reporting conventions. These features are precisely what distinguish an experimental procedure from an arbitrary physical evolution described after the fact \citep{Piccinini2025ComputationPhysicalSEP,Anderson2024}.

Reproducibility is also part of the relevant structure, although it need not mean that every run produces exactly the same raw record. In deterministic cases, repeated executions under the same conditions are expected to return the same report within the stated convention. In stochastic cases, reproducibility may instead mean that repeated executions generate a stable distribution of records or reports. The stochastic case therefore changes the form of the output object, but it does not remove the procedural structure. There is still an admissible input, a fixed protocol, a stopping rule, and a report generated according to an explicit convention. Thus stochasticity does not diminish the main goal of this paper and for clarity the fuller treatment of stochasticity is postponed until Section~\ref{sec:protocol-stochastic}.

To wrap up the main goal of this section: a fixed experiment can be understood as a finite halting procedure from admissible finite inputs to finite reports. This does not yet prove the main existence claim since the step from effective laboratory procedure to computability requires an explicit Church--Turing bridge principle. We proceed to introduce it in the next section. Once it is in place, an experiment can be treated as a finite halting procedure computing an output. Therefore, there exists a computable function at the level of its outputs.

\section{Bridge Principle and Existence of the Computed Map}
\label{sec:bridge-main-claim}
In the previous section I illustrated that a fixed experiment can be thought of as a finite halting procedure from admissible inputs to finite reports. This already gives us enough indication that my chemist colleague's original worry is misdirected. Still, to make the argument scientifically grounded and not merely rhetorical, an extra ingredient must be made explicit. The fact that a laboratory protocol is effective in practice does not by itself amount to a theorem that its report map is Turing-computable. To reach that conclusion one needs a bridge principle.

Let us recall the classical Church--Turing thesis, which will be the natural starting point for our bridge principle. In Turing's original setting, the issue is whether an effective method can be captured by a Turing computation \citep{Turing1937ComputableNumbers}. One must recall that, as Copeland emphasizes, this is not a theorem in the sense of pure mathematics, but a thesis about the correct formal analysis of effective calculability \citep{Copeland2020ChurchTuringSEP}. A laboratory protocol, however, is not a pencil-and-paper calculation. It is a physically embodied procedure carried out under controlled conditions and coupled to an explicit reporting rule. For this reason, the step from effective laboratory procedure to computable function is not automatic. We thus require a physical version of the Church--Turing idea, stated with a scope narrow enough to match the actual case at hand \citep{Piccinini2011PhysicalCTT,Schmitz2023EpistemicallyUseful}.

The bridge principle that we state below must be restricted in scope. This paper does not assume that every physical process is Turing-computable, nor that arbitrary physical evolution can always be treated as a computation. The relevant claim concerns laboratory procedures of a specific kind: procedures that are finitely specified, reproducible, equipped with admissible inputs and stopping rules, and connected to finite reporting conventions. For such procedures, the bridge principle says that their input-output behavior can be represented by a Turing-computable procedure. This is the only physical Church--Turing assumption needed here \footnote{The idea of scope limiting in this context is obviously not new and can be found, for instance, in Gandy's analysis that derives Turing computability only after imposing explicit constraints on mechanisms, or Deutsch's formulation of a physical simulation principle rather than inferring computability from physical realizability alone \citep{Gandy1980ChurchThesisMechanisms,Deutsch1985QuantumTheoryCT}.}. This restricted focus also connects the present paper with work in computability theory that models measurement as an algorithmically controlled interaction with physical equipment, while differing from that literature in using this framework to answer the existence question for protocol-level report maps \citep{BeggsCostaTucker2014,SkapinakisCosta2022}.
 
For this reason, I will use the following principle:

\medskip
\noindent\emph{Physical Church--Turing bridge principle for laboratory protocols.}  
For every fixed finitely specified, reproducible laboratory protocol \((P,C,R)\), there exists a Turing-computable procedure \(S_{P,C,R}\) such that, for every admissible finite input description \(x\) for which the laboratory protocol halts with a valid report, \(S_{P,C,R}(x)\) halts and returns exactly the same finite report as the protocol executed under the same background conditions and reporting rule.
\medskip

It is important to emphasize again that this principle is deliberately restricted. It does not say that every physically possible process is Turing-computable. It does not rule out speculative models of hypercomputation, nor does it claim that every continuum idealization used in physics is effectively simulable. Its force is narrower and more practical. It says only that ordinary laboratory protocols, understood as finite and reproducible procedures with explicit reporting conventions, do not rely on hidden hypercomputational resources at the level of their operational behaviour. If one day physics were to produce a reproducible experiment whose input-output behaviour systematically outruns all Turing computation, then this principle would fail. The present argument is therefore conditional in exactly the way announced thus far.

Once this bridge principle is stated, the main claim becomes very simple. Using the notation introduced in the previous section, the protocol executes the following:
\begin{equation}
    x \overset{E}{\longmapsto} r \overset{R}{\longmapsto} y\,,
\end{equation}
with the relevant report map being $F=R\circ E$. At this stage, the codomain is not yet, if ever, the set $\mathbb{R}$ in the idealized sense. It is the set of finite report objects actually returned by the protocol, such as rational values, intervals, finite strings with uncertainty annotations, or other finite encodings accepted by the practice. Let  $\mathrm{Rep}_{\mathrm{fin}}$ denote that set. Let $D$ denote the class of admissible finite input descriptions for which the protocol halts with a valid report. The central point of this paper can now be stated directly:

\begin{proposition}[Existence of the computed map]
Fix a laboratory protocol \((P,C,R)\) satisfying the physical Church--Turing bridge principle above. Let \(D\) be the set of admissible finite input descriptions for which the protocol terminates with a valid report, and let \(\mathrm{Rep}_{\mathrm{fin}}\) be the corresponding set of finite report objects. Then there exists a computable map
\begin{equation}
  F : D \to \mathrm{Rep}_{\mathrm{fin}}  
\end{equation}
such that, for every \(x \in D\), the value \(F(x)\) is exactly the report returned by the laboratory protocol on input \(x\).
\end{proposition}

The proof of this proposition is now clear:

\begin{proof}
By Section~\ref{sec:functions-experiments}, a fixed laboratory protocol is a finite halting procedure from admissible inputs to finite reports. The bridge principle states that there exists a Turing-computable procedure \(S_{P,C,R}\) which reproduces that same report behavior on every admissible input. For each \(x \in D\), define \(F(x)\) to be the output returned by \(S_{P,C,R}(x)\). Since \(S_{P,C,R}\) is computable, the map \(F\) is computable. Since \(S_{P,C,R}\) returns one definite report for each admissible input \(x\), the map \(F\) is a function in the minimal mathematical sense. This is exactly the existence claim at issue.
\end{proof}

The logical point is therefore straightforward. The experiment does not first presuppose a function whose existence must be justified elsewhere. Once the protocol is recognized as a halting effective procedure and the bridge principle is granted, the procedure computes a map from admissible inputs to reports. This is why the corresponding function must exist. The existence question is then not answered by some abstract metaphysical argument detached from practice. It is answered by the fact that the experiment, under the stated assumptions, computes a map.

If one prefers the standard language of computability theory, the same point may be restated by saying that the report map is a partial computable function on the ambient space of finite input descriptions, with domain \(D\). The word ``partial'' is not a defect. It simply records the familiar fact that a protocol need not apply to every imaginable input and need not return a valid report when its own stopping or quality criteria are not met. On the intended domain of admissible inputs, however, the protocol defines a computable function. That is all the existence claim requires.

It is important to clarify what we have not yet shown. The proposition does not yet identify the outputs with exact real numbers. It concerns the finite reports that the experiment actually returns. Nor does it yet say that different protocols compute the same quantity, or that stochastic protocols must be forced into a point-valued deterministic form. Those are further questions that matter, but they do not alter the core conclusion reached here.

The chemist's original worry can now be answered in the precise form needed for this paper. Once the experimental conditions, the admissible inputs, and the reporting rule are fixed, the experiment is a finite halting procedure. Under the physical Church--Turing bridge principle, that procedure computes a map. For that reason, the relevant function must exist. 

In the next section I explain how this conclusion is compatible with finite-precision measurement and how computable analysis provides the right language for moving from finite reports to real-valued idealizations when such idealizations are warranted.

\section{Finite Precision and Computable Analysis}
\label{sec:computable-analysis}
In the previous section we have established the existence of a computed map at the level of finite reports. This is the level at which an experiment directly operates. A laboratory run does not return an exact real number written with infinitely many digits. It returns a finite record, and the reporting rule converts that record into a finite numerical report, often together with an uncertainty statement and a description of the conditions under which the report is valid \citep{JCGM1002008,JCGMGUM62020}. One may think that this situation threatens the idea that experiments give access to real-valued quantities but, in fact, this is precisely the situation that requires the use of computable analysis to explain why the previous discussion applies to real-valued quantities.

Computable analysis does not require a procedure to print an exact real number in one step. A real number is computable when there is a finite procedure that, for each requested accuracy, produces a rational approximation within that accuracy. In a standard representation, a procedure computes a real number $y$ when, on input $n$, it returns a rational number $q_n$ such that
\begin{equation}
    |q_n-y|\leq 2^{-n}\,.
\end{equation}
Similarly, a real-valued function is computable when there is a procedure that, given the input and the requested accuracy $n$, returns such an approximation to the corresponding value of the function \citep{Weihrauch2000ComputableAnalysis,BrattkaHertlingWeihrauch2008Tutorial}. In this way, the computability of real-valued quantities is not opposed to finite precision but is rather formulated through finite precision.

Of course this point fits naturally with standard experimental practice. Suppose that a fixed experiment $(P,C,R)$ admits a finite refinement setting $k$. The parameter $k$ corresponds to a finite laboratory instruction internal to the protocol rather than an abstract request for the $k$-th digit of reality. Thus $k$ may specify the number of replicate measurements, an integration time, a calibration routine, an equilibration window, or a more stringent reporting rule. For each admissible input $x$ and refinement setting $k$, the protocol produces a finite report $F(x,k)=R(E(x,k))$. In the deterministic case, such a report may be represented as a rational value $q_{x,k}$, or as a pair $(q_{x,k},\Delta_{x,k})$, where $\Delta_{x,k}$ is a stated tolerance or uncertainty associated with the report.

Thus a real-valued idealization enters only when these finite reports are connected by a valid refinement rule. An experiment does not automatically deliver a unique real-valued function merely because its reports contain numbers. What is needed is a rule saying how to ask for more precision and what the resulting report warrants. If there is a real-valued quantity $f(x)$, and if the protocol supplies an effective way of choosing a refinement setting $k=k(x,n)$ such that the reported value $q_{x,k(x,n)}$ satisfies
\begin{equation}
    |q_{x,k(x,n)}-f(x)|\leq 2^{-n}\,,
\end{equation}
then the finite reports constitute a computable name for $f(x)$. In this case the experiment does not output the exact real number $f(x)$ but instead computes $f(x)$ by producing an arbitrarily accurate finite approximation of it.

This can be stated as a direct consequence of the bridge principle:

\begin{proposition}[Finite-precision realization of a real-valued computation]
Fix a deterministic laboratory protocol $(P,C,R)$ satisfying the physical Church--Turing bridge principle of Section~\ref{sec:bridge-main-claim}. Suppose that, for each admissible finite input $x$ and requested accuracy $n$, there is a computable refinement rule $k(x,n)$ such that the protocol halts at refinement setting $k(x,n)$ and returns a rational report $q_{x,k(x,n)}$ satisfying
\begin{equation}
 |q_{x,k(x,n)}-f(x)|\leq 2^{-n}   
\end{equation}
for some real-valued target $f$. Then $f$ is computable on its intended domain in the sense of computable analysis.
\end{proposition}

It is worth giving the proof.

\begin{proof}
Given $x$ and $n$, compute the refinement setting $k(x,n)$. By the bridge principle, the laboratory protocol at that refinement setting is represented by a Turing-computable procedure that returns the same finite report as the experiment. Run that procedure and extract the rational value $q_{x,k(x,n)}$. By assumption, this rational value is within $2^{-n}$ of $f(x)$. Hence there is an effective procedure that, given $x$ and $n$, returns a rational approximation to $f(x)$ with the requested accuracy. This is exactly the computable-analysis criterion for computing $f$.
\end{proof}

This proposition helps us to clarify the relation between the report-level result of Section~\ref{sec:bridge-main-claim} and the familiar language of real-valued quantities. The report map is the object that is directly computed by the experiment under the bridge principle. A real-valued function is obtained only when the finite reports can be organized into a coherent approximation scheme, so that when such a scheme is present, the concept of finite precision is not an obstacle but rather the mechanism by which the real-valued computation is expressed.

This distinction is important because laboratory uncertainty statements do not always function as deterministic error bounds. What I mean is that in many contexts, an uncertainty statement summarizes a measurement model, a coverage convention, or a statistical statement of dispersion \citep{JCGM1002008,JCGM1012008}.  In those cases, the report may not justify an assertion of the form \(|q-f(x)|\leq 2^{-n}\) for each individual run. The deterministic proposition above should therefore be read as the clean limiting case. Thus, stochastic experiments require a corresponding treatment in terms of output distributions, sampling procedures, and assumptions about the stability of those distributions. This extra ingredient of stochasticity will be introduced and treated in Section~\ref{sec:protocol-stochastic}.

Going back to the solubility example, which worried my chemist colleague enormously, allows us to illustrate the discussion so far more clearly. A refined solubility protocol may increase equilibration time, improve temperature control, repeat assays, or tighten calibration. These changes may support increasingly accurate reports, but only if the protocol itself specifies how refinement is to be interpreted and how the resulting reports are warranted. If the refinement structure is coherent, then the finite reports approximate a real-valued solubility function in the computable-analysis sense. If it is not coherent, the experiment still computes a finite report map, but the stronger claim that there is a unique real-valued solubility function has not yet been earned.

To summarise, no experiment has to print infinitely many digits in order to compute a real-valued quantity. What is required is a finite procedure that can answer finite accuracy requests in a controlled way. We have shown that computable analysis is the correct mathematical language to capture this feature, because it preserves the central conclusion of the paper while keeping the level of the claim precise and incorporating the realistic ingredient of finite-precision measurement. The experiment computes a finite report map. Under suitable refinement guarantees, that report map also gives effective access to a real-valued function.

\section{Protocol Dependence and Stochastic Qualification}
\label{sec:protocol-stochastic}
\label{sec:operational-measurands}
\label{sec:stochastic}
The main result is now in place: a fixed experiment is a finite halting procedure, which, under the physical Church--Turing bridge principle, computes a map from admissible inputs to reports. In this section I am going to add two extra ingredients which are important in scientific practice but do not change the logic of the argument. The first ingredient is that the map is fixed by the protocol. The second one is that, in many cases, the output of the protocol is stochastic. Protocol dependence tells us which map is computed while stochasticity changes the form of the output object to a distribution. Neither point undermines the existence claim we have established before.

It is useful at this stage to denote the protocol-relative target explicitly. For a fixed protocol $P$, background conditions $C$, and reporting rule $R$, I will call the corresponding target an operational measurand, which is, rather simply, a compact way of saying that the quantity under discussion is fixed by the way the experiment is specified and reported. This is consistent with the metrological idea that specifying a measurand is part of the measurement task, and that the measurement result depends on the model and the reporting conventions used to obtain it \citep{JCGM2002012,JCGM1002008,JCGMGUM62020}. In analytical measurement, this dependence is especially explicit in method-defined cases, where the measurand is defined by reference to the procedure used to obtain it \citep{BarwickEurachemTAM2023}.

Thus the notation $F=R\circ E$ should be read more fully as
\begin{equation}
    F_{P,C,R}=R\circ E_{P,C}\,.
\end{equation}

Changing $P$, $C$ or $R$ may change the computed map. Thus, more precisely, once these specifications are fixed, the corresponding experimental procedure computes a map, and the relevant function exists at that level. This is not the end of the scientific task. In many contexts, scientists also want to know whether maps obtained under different protocols, background conditions, or reporting rules can be treated as measurements of the same underlying quantity. Such protocol-independent or protocol-robust claims require additional arguments from calibration, correction models, interlaboratory comparison, robustness, theoretical modeling, coherence testing, or convergence under refinement \citep{JCGM1002008,Tal2016,Bokulich2020CalibrationCoherence}. They are important scientific achievements, but they should not be built into the existence question from the beginning.

The solubility example makes this point concrete. A solubility protocol that fixes one solvent, one temperature regime, one equilibration criterion, and one reporting convention computes one report map. A protocol that changes the solvent, the endpoint criterion, the separation step, or the analytic readout may compute a different report map. It is then misleading to say, without further argument, that one procedure measures the real solubility while the other merely introduces noise. What is true is more precise. Each sufficiently specified procedure fixes its own operational target. Whether two such targets can be treated as measurements of the same underlying quantity is a further and deep scientific question. Lack of immediate protocol independence does not imply lack of function. It tells us that the function whose existence has been secured is the one computed by the fixed experimental procedure.

The second ingredient concerns stochasticity. Up to now we have considered the cleanest case corresponding to deterministic maps: for a fixed admissible input $x$, and perhaps for a refinement setting $k$, the protocol returns a report $F(x,k)$. Many experiments, however, do not behave like this. Repeated executions under the same protocol and the same background conditions may return different finite records and hence different reports. This variability may come from instrument noise, sampling variability, microscopic initial conditions, preparation differences, or other sources that the protocol does not eliminate. In such cases, reproducibility does not mean identical reports in every run but rather means stability of the report-generating distribution under repeated execution.

The mathematical object is then not a single report but a distribution over possible reports. Let $\omega$ denote the random element involved in one execution of the protocol. The execution stage may be written as $E(x,k,\omega)$ and the reported output as:
\begin{equation}
    Y_{x,k}=R(E(x,k,\omega))\,.
\end{equation}
For each admissible pair $(x,k)$, this induces a probability distribution $\mu_{x,k}$ on the finite report space $\mathrm{Rep}_{\mathrm{fin}}$ defined by
\begin{equation}
    \mu_{x,k}(A)=\text{Prob}\left(Y_{x,k}\in A\right)\,,\qquad A\subseteq \mathrm{Rep}_{\mathrm{fin}}\,.
\end{equation}
The stochastic target is therefore the family
\begin{equation}
    (x,k)\longmapsto \mu_{x,k}\,.
\end{equation}
Point estimates, confidence intervals, posterior summaries, or sample means are then obtained by applying reporting conventions or summary rules to this distribution. 

This formulation is consistent with ordinary measurement practice. In the GUM framework, uncertainty may be represented by probability distributions assigned to input quantities and propagated through a measurement model \citep{JCGM1002008}. GUM Supplement~1, for instance, makes this sampling perspective explicit through Monte Carlo propagation of distributions \citep{JCGM1012008}. The important point for the present paper is not that every laboratory report must be written in this language but rather that when repeated execution naturally produces variable reports, the correct operational target is distributional rather than point-valued.

In this more natural scenario for an experiment, the computability claim should then be stated at this level of description. In the stochastic case, the bridge principle is not  used to obtain a deterministic function returning one fixed report. It is used to obtain a computable sampling procedure for the induced report distribution. More explicitly, under the stochastic reading of the same restricted bridge principle, there is a procedure $S_{P,C,R}$ such that, for admissible $(x,k)$, and for a source of randomization represented by $\omega$, we have that
\begin{equation}
    S_{P,C,R}(x,k,\omega)\sim \mu_{x,k}\,.
\end{equation}

This means that the algorithmic object is a sampler for the report distribution. The random input $\omega$ is not a trick that hides the problem. It marks the fact that the protocol itself is being represented as a randomized effective procedure rather than as a deterministic one.

This is the stochastic analogue of the deterministic report map. In the deterministic case, the protocol computes
\begin{equation}
    F_{P,C,R}(x,k)\in \mathrm{Rep}_{\mathrm{fin}}\,.
\end{equation}
In the stochastic case, the protocol computes, in the operational sense, a sampling procedure for
\begin{equation}
    \mu_{x,k}\in \mathcal{P}(\mathrm{Rep}_{\mathrm{fin}})\,.
\end{equation}
The existential conclusion remains the same since the experiment does not fail to define a function because its individual runs vary. It defines a stochastic map from admissible inputs to report distributions, or equivalently a randomized procedure that generates reports according to those distributions.

There is one important restriction. Sampling access to a distribution is not the same thing as symbolic access to all of its probabilities, moments, densities, or expectations. Repeated execution gives samples from the induced report distribution. It does not automatically give closed-form probabilities or arbitrary integrals of test functions. In computable analysis and computable probability, different representations of a measure support different kinds of computation \citep{Weihrauch2000ComputableAnalysis,HoyrupRojas2009ComputabilityMeasures}. The representation closest to experimental practice is sampling. Stronger representations may be available in particular cases, but they require further assumptions.

This also clarifies the relation between stochasticity and the real-valued idealizations discussed in Section~\ref{sec:computable-analysis}. A deterministic real-valued target requires a refinement structure that yields approximations to a single value. A stochastic target requires a refinement structure that stabilizes a family of report distributions. If those distributions converge effectively to a limiting distribution in a chosen representation, then one may obtain a stronger distribution-valued computability claim. If instead the protocol only stabilizes distributions at a declared operational resolution, then the correct object remains the protocol-level family $\mu_{x,k}$. This is not a defect but rather the correct expression of what the experiment actually computes.

The solubility example again illustrates this point clearly. Repeated executions of a solubility protocol may not return the same reported concentration. Differences in sample handling, equilibration judgment, instrument response, and assay variability may produce a spread of reports even when the protocol is followed correctly. If that spread is stable under the stated conditions, then the protocol fixes a report distribution. A reported mean, interval, or uncertainty statement is a summary of that distribution according to $R$. The existence claim is not defeated by variability as the computed object becomes a stochastic one.

Thus, protocol dependence and stochasticity are not objections against the paper's existence claim but rather add realistic ingredients to it. Protocol dependence specifies the experimental map whose existence is at issue while stochasticity specifies whether the output of that map is a finite report, a distribution of reports, or a sampler for such a distribution. Therefore the main argument of the existence claim survives unchanged and becomes more precise. A fixed experiment is a finite halting procedure which, under the bridge principle, is computable. In the deterministic case it computes a report map, while in the stochastic case it computes a sampling procedure for a report distribution.

In the next section we move to state the limits of our claims. Obviously, stronger claims about a single protocol-independent real-valued function require additional coherence, calibration, and invariance assumptions.

\section{Limits and Non-Goals}
\label{sec:limitations}

We have now reached the existence claim in its intended form, so it is worth stating it once more. A fixed experiment is a finite halting procedure. Under the physical Church--Turing bridge principle, that procedure computes a map from admissible inputs to reports. In the deterministic case this gives a computable report map, and under suitable refinement guarantees it gives effective access to a real-valued function. On the other hand, in stochastic cases it gives a computable sampling procedure for the induced report distribution. These conclusions are now enough to answer my chemist colleague's existence worry that motivated me to write this paper. They are not meant to settle every question about physics, measurement, or computation. It is therefore useful to enumerate clearly what has not been claimed.

First, the paper does not prove nor does it claim to prove a universal physical Church--Turing thesis, even a restricted one. The bridge principle here is restricted to finitely specified, reproducible laboratory protocols with admissible inputs, stopping rules, and finite reporting conventions. It is not a claim that every physical process is Turing-computable, nor that arbitrary physical evolution can always be represented as an algorithm. The bridge principle is narrower: ordinary experiments have a procedural structure that makes them appropriate objects for computability analysis and classification once the bridge principle is accepted. This latter statement is important. Bluntly, the arguments presented here are conditional on that principle rather than pretending to derive it from pure mathematics \citep{Turing1937ComputableNumbers,Copeland2020ChurchTuringSEP,Piccinini2011PhysicalCTT}. If one day science were to produce a reproducible laboratory protocol whose input-output behavior systematically exceeded Turing computability, then the bridge principle adopted here would fail for that case. The present argument does not rule out that possibility by definition. It only says what follows for experimental procedures falling under the stated bridge principle.

Second, the paper does not claim that every protocol computes the same protocol-independent quantity. On the contrary, Section~\ref{sec:protocol-stochastic} emphasized that changing the protocol, the background conditions, or the reporting rule may change the computed map. This is not a weakness in the existence claim. It is the correct location of the claim. The function whose existence has been secured is the one computed by the fixed experimental procedure. Whether different procedures can be treated as measurements of the same underlying quantity is a further and important scientific question. That further question is  central to science. We often want different instruments, laboratories, sample preparations, or reporting conventions to converge on a common quantity. When they do, the result is important. But such agreement has to be earned through calibration, correction models, robustness, interlaboratory comparison, theoretical explanation, or convergence under refinement \citep{JCGM1002008,Tal2016}. Protocol independence is therefore not a precondition for the existence of a function. It is a higher-level achievement about relations among functions computed by different protocols.

Third, the paper does not claim that every finite report map automatically yields a unique exact real-valued function. The report map is primary because it is what the experiment directly computes. A real-valued idealization requires an additional refinement structure showing how finite reports approximate a target value to requested accuracy. Without such a structure, there may still be a perfectly good computed report map, but the stronger claim that this map gives effective access to a unique real-valued function has not been justified. This is why Section~\ref{sec:computable-analysis} treated real-valued computability as an additional layer rather than as the starting point.

Fourth, a stochastic experiment does not fail to define a function because repeated runs vary. Rather, the computed object is different. It may be a sampler for a report distribution, or a map from admissible inputs to distributions over reports. What is not automatic is symbolic access to all probabilities, expectations, moments, or densities associated with that distribution. Sampling access is already a meaningful computability claim, but stronger distributional claims require stronger representations and additional assumptions \citep{Weihrauch2000ComputableAnalysis,HoyrupRojas2009ComputabilityMeasures}.

Finally, the paper does not make a complexity claim. To say that a report map is computable is not to say that it is efficiently computable, practically feasible, or easy to approximate. A laboratory procedure may halt only after a very long time, require expensive apparatus, or demand a refinement schedule whose cost grows rapidly with the requested accuracy. Those issues matter for scientific practice, but they belong to complexity, resource analysis, and experimental design. They do not affect the basic existence claim made here. A map may exist and be computable while still being difficult to evaluate in practice.

The limits just stated are not qualifications added to weaken the result. They are part of what makes the result presented here precise. The paper does not need to show that all physical processes compute, that all protocols converge to one universal quantity, or that all computable maps are efficiently accessible. Its claim is sharper than that: a reproducible experiment is a finite halting procedure. Under the stated bridge principle, it computes a map. Therefore the relevant function exists at the level fixed by the protocol, the background conditions, and the reporting rule. The remaining work is to explain that function, compare it with functions obtained under other protocols, and determine when stronger claims about real-valued targets or protocol-independent quantities are warranted.

\section{Conclusion}
\label{sec:conclusion}
\label{sec:payoff}

This paper began with a simple question that came from experimental practice rather than from an abstract puzzle. Does a quantity such as solubility define a function? The answer defended here is yes, once the question is asked at the right level. A reproducible experiment is a finite procedure for obtaining a report. It takes admissible inputs, runs through a specified sequence of operations, halts according to a stopping rule, and returns an output through a reporting convention. Under the physical Church--Turing bridge principle stated in Section~\ref{sec:bridge-main-claim}, such a procedure computes a map. For that reason, the corresponding function exists at the level fixed by the protocol, the background conditions, and the reporting rule.

This conclusion does not require that the experiment print an exact real number. The report produced in the laboratory is finite. It may be a rational value, an interval, a finite string, an uncertainty statement, or a more elaborate finite record. Computable analysis is precisely the language that allows us to connect such finite reports with real-valued quantities when a valid refinement structure is available \citep{Weihrauch2000ComputableAnalysis,BrattkaHertlingWeihrauch2008Tutorial}. The important point is not exact access to infinitely many digits in one step. The important point is whether the procedure can answer finite accuracy requests in a controlled way. When it can, the finite report map gives effective access to a real-valued function. When it cannot, the finite report map remains the object that the experiment directly computes.

The same conclusion survives protocol dependence and stochasticity. Changing the protocol, the background conditions, or the reporting rule may change the computed map. This is not a failure of the argument. It tells us which function has been secured. A fixed protocol computes its own operational map, and the question whether different protocols should be treated as measurements of the same underlying quantity is a further scientific question. In stochastic settings, repeated runs may generate a stable distribution of reports rather than a single deterministic output. This also does not destroy the existence claim. It changes the computed object from a point-valued report map to a sampler, or to a map into report distributions. The logic remains the same.

The point of the paper is therefore not to prove a universal physical Church--Turing thesis. It is also not to guarantee that all experimental protocols converge to a single protocol-independent real-valued quantity. Those would be much stronger claims, and the preceding section explained why they should not be built into the argument. The result here is sharper. Once an experiment is specified as a finite halting procedure, and once the bridge principle is accepted, the experiment computes a map. The existence of the relevant function follows from that computation.

This reframes the scientific problem. The central difficulty is not that experimentally defined quantities fail to exist because they depend on procedures. The central difficulty is to understand the functions that experiments compute. One must explain their structure, predict their values, compare them across protocols, and determine when different operational maps can be unified as measurements of the same underlying quantity. Such unification may require calibration, robustness arguments, correction models, theoretical explanation, or model selection \citep{JCGM1002008,Tal2016,Woodward2003MakingThingsHappen,Levins1966StrategyModelBuilding,Wimsatt1981RobustnessReliability,Akaike1974AIC,BurnhamAnderson2002ModelSelection}. These are substantive scientific achievements. They begin after the existence question has been put in its proper form.

The original worry of my chemist colleague can now be answered directly. If one looks for a solubility function that is already independent of solvent, temperature, equilibration criteria, sample preparation, and reporting convention, then the question is badly posed. But if one fixes the experimental conditions and asks what the experiment does, the answer is clear. The experiment is a finite procedure that halts and reports a value. Under the bridge principle, it computes a map. Therefore the solubility function, at that operational level, exists. The deeper scientific work is not to rescue that existence from nothing, but to explain the computed map and to understand how it relates to other maps obtained under other protocols.

\bmhead{Acknowledgements}
I would like to thank an anonymous chemist colleague and friend for a conversation over lunch in which he pressed a worry that initially sounded almost metaphysical but turned out to be methodologically revealing: whether a quantity such as solubility ``defines a function'' in any strict mathematical sense. My initial attempt to reassure him by treating an experiment as a finite sequence of steps taking a specified input situation to a recorded outcome quickly exposed a deeper issue, namely that the real question was not simply what one is entitled to call a function, but what the relevant object of analysis is in the first place. That conversation led me to the view developed here: that the right starting point is not a protocol-independent property function, but the operational map fixed by a protocol, background conditions, and a reporting convention. The present paper grew out of that disagreement and out of the attempt to reconcile the minimal set-theoretic notion of a function with the protocol dependence and finite precision that are constitutive of experimental practice.

\bibliography{biblio}

\end{document}